\documentclass[svgnames]{llncs}
\usepackage{stmaryrd}
\usepackage{amsmath,amssymb,amsfonts,mathrsfs,thm-restate,thmtools}
\usepackage{graphicx}
\usepackage{appendix}
\usepackage{epsf}
\usepackage{epsfig}
\usepackage{epstopdf}
\usepackage{xspace,xcolor}
\usepackage{soul}
\usepackage{latexsym}
\usepackage{tikz}
\usepackage{framed}

\newtheorem{Reduction}{Reduction}
\usepackage[backgroundcolor=Moccasin,colorinlistoftodos,obeyDraft]{todonotes}

\makeatletter
 \providecommand\@dotsep{5}
 \def\listtodoname{}
 \def\listoftodos{\@starttoc{tdo}\listtodoname}
 \makeatother

\newcounter{nmcomment}

\usepackage{caption}
\usepackage{subcaption}
\usepackage{tkz-tab}
\usetikzlibrary{arrows,positioning}
\usetikzlibrary{shapes,snakes}
\usepackage[linesnumbered,ruled,vlined]{algorithm2e}
\SetCommentSty{mycommfont}
\usepackage[linkbordercolor={0 0 1}]{hyperref}
\usepackage{nameref}

\usepackage[ruled,vlined]{algorithm2e}
\usepackage{enumerate}
\usepackage{caption}
\usepackage{subcaption}
\usepackage[compatibility=false]{caption}
\tikzstyle{vertex}=[circle, draw, inner sep=0pt, minimum size=4.5pt]

\newcommand{\NP}{\ensuremath{\sf{NP}}\xspace}

\newcommand{\fpt}{\ensuremath{\sf{FPT}}\xspace}

\newcommand{\cfcn}{\textsc {CF-CN}}
\newcommand{\cfon}{\textsc {CF-ON}}

\newcommand{\problembox}[4]{
\begin{framed}
{\sc #1} \\
\begin{tabular}{p{.18\textwidth} p{.8\textwidth}}
\hfill \bf Input: & {#2}\\
\hfill \bf Parameter: & {#3}\\
\hfill \bf Question: & {#4}\\
\end{tabular}
\end{framed}
}

\begin{document}

\title{Parameterized Algorithms for Conflict-free Colorings of Graphs}
\titlerunning{The Parameterized Complexity of Conflict-free Colorings}  
%
\author{I. Vinod Reddy}
\authorrunning{I. V Reddy} 
%

%

\institute{IIT Gandhinagar, India \\
 \email{reddy\_vinod@iitgn.ac.in}
 }

\maketitle              

\begin{abstract}
In this paper, we study the conflict-free coloring of graphs induced by neighborhoods.
A coloring of a graph is conflict-free if every vertex has a uniquely colored vertex in its neighborhood. The conflict-free coloring problem is to color the vertices of a graph using the minimum number of colors such that the coloring is conflict-free. 
We consider both closed neighborhoods, where the neighborhood of a vertex includes 
itself, and open neighborhoods, where a vertex does not included in its
neighborhood. We study the parameterized complexity of conflict-free closed neighborhood coloring and conflict-free open neighborhood coloring problems. We show that both problems are fixed-parameter tractable ($\fpt$) when parameterized by the cluster vertex deletion number of the input graph. This generalizes the result of Gargano et al.(2015) that conflict-free coloring is fixed-parameter tractable parameterized by the vertex cover number. Also, we show that both problems admit an additive constant approximation algorithm when parameterized by the distance to threshold graphs.

We also study the complexity of the problem on special graph classes. We show that both problems can be solved in polynomial time on cographs. For split graphs, we give a polynomial time algorithm for closed neighborhood conflict-free coloring problem, whereas we show that open neighborhood conflict-free coloring is $\NP$-complete.
We show that interval graphs can be conflict-free colored using at most four colors.  
\end{abstract}

\section{Introduction}
A hypergraph is a pair $ \mathcal{H}=(V, E)$ where $V$ is the set of vertices and $E$ is a set of non-empty subsets of $V$ called \emph{hyperedges}.  
A proper $k$-coloring of $\mathcal{H}$ is an assignment of colors from $\{1,2, \cdots, k\}$ to every vertex of $\mathcal{H}$ such that every hyperedge contains at least two vertices of distinct colors.  
The minimum number of colors required to color the vertices of
$\mathcal{H}$ is called the chromatic number of $\mathcal{H}$ and is denoted as $\chi( \mathcal{H})$.
The \emph{conflict-free} coloring is a special case of coloring and it is defined as follows.
\begin{definition}
 Let $ \mathcal{H}=(V, E)$ be a hypergraph, a coloring $C_\mathcal{H}$ is called conflict-free coloring of $\mathcal{H}$ if for every $e \in E$
there exists a vertex $u \in e$ such that for all $v\in e$, $u \neq v$ we have $C_\mathcal{H}(u) \neq C_\mathcal{H}(v )$.  
The minimum number of colors needed to conflict-free color the vertices of a
hypergraph $\mathcal{H}$ is called the \emph{conflict-free chromatic number} of $\mathcal{H}$.
\end{definition}
The conflict-free coloring problem was introduced by Even et al. \cite{even2003conflict}   
to study the frequency assignment problem for cellular networks.
These networks contain two types of nodes, base stations and clients. Fixed frequencies are assigned to base stations to allow connections to clients. Each client scans for the available base stations in this neighborhood and connects to one of the available base station. 
Suppose if two base stations are available to a client, which are assigned the same frequency then mutual interference occurs and the connection between the client and base stations can become noisy. Our aim is to reduce the disturbances occur in connections between base stations and clients. The frequency assignment problem on cellular networks is an assignment of frequencies to base stations such that for 
each client there exists a base station of unique frequency within his region. The
goal here is to minimize the number of assigned frequencies, since available frequencies are limited and expensive. 

We can model this problem using the hypergraphs. The vertices of hypergraph correspond to the base stations and the set of base stations available for each client is represented by a hyperedge. The problem reduces to assigning frequencies to vertices of hypergraph such that each hyperedge contains a vertex of unique frequency.   
Conflict-free coloring is well studied for hypergraphs induced by geometric objects like, intervals \cite{bar2008deterministic}, rectangles \cite{ajwani2007conflict}, unit disks \cite{lev2009conflict} etc. 
This problem also has applications in areas like radio frequency identification and robotics, VLSI design and many other fields.

In this paper, we study the 
conflict-free coloring of hypergraphs induced by graph neighborhoods.
Let $G=(V,E)$ be a graph, for a vertex $v \in V(G)$, $N(v)$ denotes the set consisting of all vertices which are adjacent to $v$, called open neighborhood of $v$. The set $N[v]=N(v) \cup \{v\}$ is called the closed neighborhood of $v$.
The conflict-free open neighborhood (\cfon{}) coloring of a graph $G$ is defined as the conflict-free coloring of the hypergraph $\mathcal{H}$ with 
$$ V(\mathcal{H}) =V(G) \qquad \textrm{and} \qquad E(\mathcal{H})=\{N(v)~:~ v \in V(G)\}$$
Similarly conflict-free closed neighborhood (\cfcn{}) coloring  problem can be defined. 

Alternatively we can also define both \cfcn{} coloring and \cfon{} coloring problems as follows.
Given a graph $G$ and a coloring $C_G$, we say that a subset $U \subseteq V(G)$ has a \emph {unique color} with respect to $C_G$ if there exists a color
$c$ such that $|\{u \in U ~|~ C_G(u)=c\}|=1$.

\begin{definition} 
\begin{enumerate}
 \item  A coloring $C_G$ of a graph $G$ is called \emph{conflict-free closed neighborhood (\cfcn{}) coloring}  if for every vertex $v \in V(G)$, the set $N[v]$ has a unique color.  
 \item  A coloring $C_G$ of a graph $G$ is called \emph{conflict-free open neighborhood (\cfon{}) coloring}  if for every vertex $v\in V(G)$, the set $N(v)$ has a unique color.  
\end{enumerate}
 \end{definition}
The minimum value $k$ for which there is a \cfon{} (resp. \cfcn{}) coloring of $G$
with $k$ colors is called the \emph{\cfon{} (resp. \cfcn{}) chromatic number} of $G$ and is denoted as $\chi_{cf}(G)$ (resp.  $\chi_{cf}[G]$).

\vspace{-0.1cm}
\begin{figure}
\centering
 \includegraphics[scale=0.6,trim={3cm 15.5cm 2cm 3.2cm},clip]{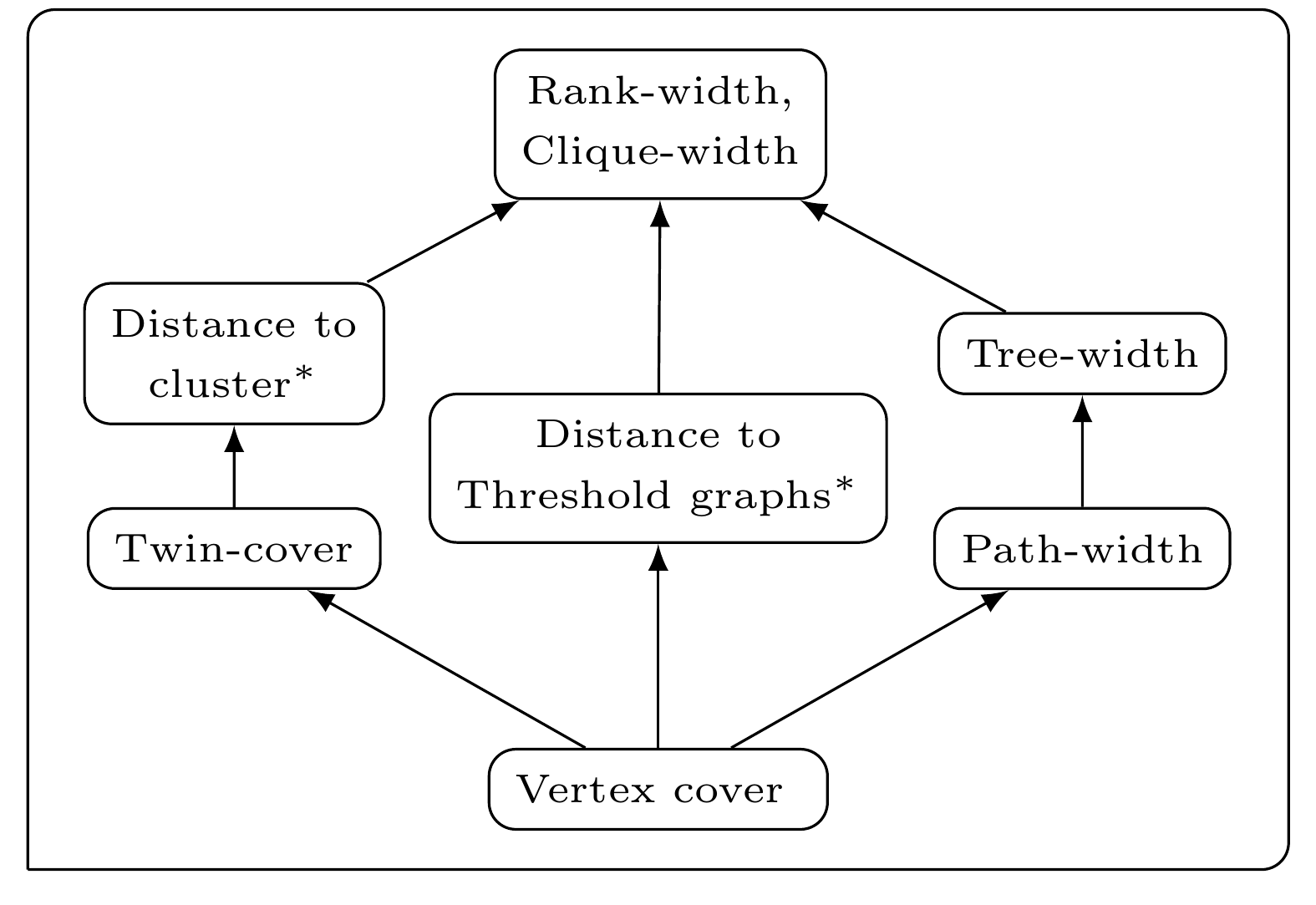}
 \caption{~A schematic showing the relation between the various parameters. An arrow from parameter $a$ to $b$ indicates that $a$ is larger than $b$. Parameters marked with $*$ are studied in this paper.}
\label{fig-parameters}
\end{figure}
\paragraph{ Related work.} Gargano et al.\cite{gargano2015complexity} studied the complexity of conflict-free colorings induced by the graph neighborhoods and showed that
the \cfcn{} 2-coloring and \cfon{} $k$-coloring are $\NP$-complete.
In the parameterized setting, both conflict-free closed and open neighborhood colorings are fixed-parameter tractable ($\fpt$), when parameterized by the vertex cover number or the neighborhood diversity of the graph \cite{gargano2015complexity}.
Ashok et al. \cite{ashok2015exact} showed that maximizing the number of conflict-free colored edges in hypergraphs is $\fpt$ when parameterized by the number of conflict-free edges in  solution.

\paragraph{ Our contributions.} 
In this paper, we give parameterized algorithms for \cfcn{} coloring and \cfon{} coloring problems with respect to various structural parameters. Gargano et al.\cite{gargano2015complexity} showed that both problems are $\fpt$ parameterized by vertex cover number. Both problems are $\fpt$ parameterized by tree-width, which follows from an application of Courcelle's Theorem \cite{courcelle1990monadic} and the fact that the \cfcn{} and \cfon{} problems can be expressed by a monadic second order (MSO) formula. 

In this paper, we focus is on \emph{distance-to-triviality} ~\cite{guo2004structural,cai2003parameterized} parameters. They measure how far a graph is from some class of graphs for which the problem is tractable.
Then, it is natural to parameterize by the distance of a general instance to a tractable class.
The main advantage of studying structural parameters is, if a problem is tractable on a class of graphs $\mathcal{F}$, 
then it is natural to expect the problem might be tractable on a class of graphs which are close to $\mathcal{F}$.  
Our notion of distance to a graph class is the vertex deletion distance. More precisely, for a class $\mathcal{F}$ of graphs we say that $X$ is a $\mathcal{F}$-modulator of a graph $G$ if there is a  subset $X\subseteq V(G)$ such that  $G \setminus X \in \mathcal{F}$. If the size of the smallest modulator to $\mathcal{F}$ is $k$, we also say that the distance of $G$ to the class $\mathcal{F}$ is $k$.

We study the parameterized complexity of the conflict-free coloring problems with respect to the distance from following graph class: \emph{cluster graphs} (disjoint union of cliques) and \emph{threshold graphs}. 
Studying the parameterized complexity of conflict-free coloring problem with respect to these parameters improves our understanding about the tractable parameterizations. 
For instance, the parameterization by the distance to cluster graphs directly generalizes vertex cover and is not comparable with tree-width (see Fig.~\ref{fig-parameters}). 
In particular we obtain the following results.
\vspace{-0.2cm}
\begin{itemize}
\item We show that both variants of conflict-free coloring problems are \fpt{} when parameterized by size of the modulator to cluster graphs (cluster vertex deletion number).
\item We show that \cfcn{} (resp. \cfon{}) coloring problem admits an additive $1$-approximation (resp. $2$-approximation) when parameterized by the size of the modulator to threshold graphs.
\item We show that on split graphs \cfcn{} coloring admits a polynomial time algorithm, where as \cfon{} coloring is $\NP$-complete.
\item We give polynomial time algorithms for both  \cfcn{} and \cfon{} coloring problems on cographs.
We show that interval graphs can be conflict-free colored using at most four colors.  
\end{itemize}

\section{Preliminaries}

In this section, we introduce the notation and the terminology that we will need to describe our algorithms. Most of our notation is standard. We use $[k]$ to denote the set $\{1,2,\ldots,k\}$. All graphs we consider in this paper are undirected, connected, finite and simple.
For a  graph $G=(V,E)$, let $V(G)$ and $E(G)$ denote the vertex set and edge set of $G$ respectively.
An edge in $E$ between vertices $x$ and $y$ is denoted as $xy$ for simplicity. 
For a  subset $X \subseteq V(G)$, the graph $G[X]$ denotes the subgraph of $G$ induced by vertices of $X$. Also, we abuse notation and use $G \setminus X$ to refer to the graph obtained from $G$ after removing the vertex set $X$. 
For a vertex $v\in V(G)$,
$N(v)$ denotes the set of vertices adjacent to $v$ and $N[v] = N(v) \cup \{v\}$ is the closed neighborhood of $v$. 
A vertex is called \emph{universal vertex} if it is adjacent to every other vertex of the graph.

\paragraph{Graph classes.} We now define the graph classes which are considered in this paper. 
 A graph is a \emph {split graph} if its vertices can be partitioned into a clique and an independent set.  Split graphs are $(2K_2,C_4,C_5)$-free. 
 The class of $P_4$-free graphs are called \emph{cographs}. 
 A graph is a \emph{threshold graph} if it can be constructed recursively by adding an isolated vertex or a universal vertex.
 A cluster graph is a disjoint union of complete graphs. Cluster graphs are $P_3$-free graphs.
 A graph $G$ is called an \emph{interval graph} if there exists a set $\{I_ v ~|~ v \in V(G) \}$ of
real intervals such that $I_u \cap I_v \neq \emptyset$  if and only if $(u,v) \in E(G)$.

 It is easy to see that a graph that is both split and cograph is a threshold graph. We denote  threshold graph (or a split graph) with $G = (C,I)$ where $C$ and $I$ denotes the partition
of $G$ into a clique and an independent set. For any two vertices $x, y$ in a threshold graph $G$ we have either $N(x)\subseteq N[y]$ or $N(y)\subseteq N[x]$. For a class of graphs $\mathcal{F}$, the distance to $\mathcal{F}$ of a graph
$G$ is the minimum number of vertices to be deleted from $G$ to get a graph in $\mathcal{F}$. 
An $\mathcal{F}$-modulator is the smallest possible
vertex subset $X$ for which $G \setminus X \in \mathcal{F}$. The size of a $\mathcal{F}$-modulator is a natural measure
of closeness to $\mathcal{F}$.

 \paragraph{Parameterized Complexity.} A parameterized problem denoted as $(I,k)\subseteq \Sigma^*\times \mathbb{N}$, where $\Sigma$ is fixed alphabet and $k$ is called the parameter. We say that the problem $(I,k)$ is {\it fixed parameter tractable} with respect to parameter $k$ if there exists an algorithm which solves the problem in time $f(k) |I|^{O(1)}$, where $f$ is a computable function. A kernel for a parameterized problem $\Pi$ is an algorithm which transforms an instance $(I,k)$ of $\Pi$ to an equivalent instance $(I',k')$ in polynomial time such that $k' \leq g(k)$ and $|I'| \leq f(k)$  for some computable functions $f$ and $g$. It is known that a parameterized problem is fixed parameter tractable if and only of it has a kernel. For a detailed survey of the methods used in parameterized complexity, we refer the reader to the texts \cite{CyganFKLMPPS15,downey2013fundamentals}.

Since cluster graphs are $P_3$-free and threshold graphs are $(P_4, C_4, 2K_2)$-free,  modulators can be computed in \fpt{} time. 
Therefore we assume that a modulator to these graph classes is given as a part of the input. 
The problems we consider in this paper are formally defined as follows:
We only define for closed neighborhood case and open neighborhood case can be defined similarly.
 
\problembox{\sc Conflict-free closed neighborhood coloring (CF-CNC)}{A graph $G$, a subset $X \subseteq V(G)$ such that $G \setminus X \in \mathcal{F}$ and an integer $k$.} {The size $d:=|X|$ of the modulator to $\mathcal{F}$.}{Does there exists a coloring $C_G: V(G) \rightarrow [k]$ such that for each $v \in V(G)$ the set $N[v]$ has a unique color with respect to $C_G$?}

\section{Parameterized Algorithms}
In this section, we give parameterized algorithms for conflict-free open/closed neighborhood coloring problems parameterized by cluster vertex deletion number and distance to threshold graphs.
\subsection{Parameterized by Cluster Vertex Deletion Number}
The cluster vertex deletion number (or distance to cluster) of a graph $G$ is the minimum number of vertices
that have to be deleted from $G$ to get a disjoint union of complete graphs or
cluster graph. Cluster vertex deletion number is
an intermediate parameter between vertex cover number and clique-width/rank-width \cite{doucha2012cluster}.
In this section we show that both variants of conflict free coloring problems are $\fpt$ parameterized by cluster vertex deletion number. The following lemma gives an upper bound on number of colors needed in a conflict-free coloring.
\begin{lemma}\label{lem-bound}
Let $X \subseteq V(G)$ of size $d$ such that $G \setminus X$ is a cluster graph then
$$\chi_{cf}[G] \leq d+2 \qquad \textrm {and} \qquad \chi_{cf}(G) \leq 2d+2$$
\end{lemma}
\begin{proof}
A \cfcn{} $(d+2)$-coloring $C_G$ of $G$ can be obtained as follows.
\begin{enumerate}
\item For each clique $C \in G \setminus X$, assign $C_G(u)=0$ for some vertex $u \in C$ and  for all 
$v \in C \setminus \{u\}$, assign color $C_G(v)=1$
\item For each $x \in X$ assign $C_G(x)$ a color from the set $\{2, \cdots, d+1\}$ that is not already used by $C_G$.
\end{enumerate}
According to this coloring, the unique color in $N[x]$ is $C_G(x)$ if $x \in X$ and $0$ if $x \in G \setminus X$.

A \cfon{} $(2d+2)$-coloring $C_G$ of $G$ can be obtained as follows.
\begin{enumerate}
\item For each vertex $u \in G \setminus X$, assign the color $C_G(u)=0$ 
\item For each $x \in X$ assign $C_G(x) \in \{1, \cdots, d\}$ that is not already used by $C_G$.
\item For a vertex $x \in X$ if $C_G(N(x))= \{0\}$, then recolor any one vertex in $N(x)$ with a color from $\{d+1, \cdots ,2d\}$ which is not already used by $C_G$.

\item For each clique $C \in G \setminus X$, if $C_G(C)=\{0\}$, then recolor an arbitrary vertex $u$ in $C$ with color $2d+1$.
\end{enumerate}
It is easy to see that above coloring is a \cfon{} coloring of $G$. \qed
\end{proof}

\begin{theorem} \label{th-dcluster}
Both \cfcn{} and \cfon{} coloring problems are fixed-parameter tractable when parameterized by the cluster vertex deletion number of the input graph.
 \end{theorem}
\begin{proof}
Let $G$ be a graph and $X \subseteq V(G)$ of size $d$ such that $G \setminus X$ is a cluster graph. 
\par {\bf \cfcn{} coloring} :
First we show that \cfcn{} coloring problem admits a kernel of size $O(d^{2^d+1})$.
Without loss of generality we assume that $k < d+2$ otherwise from Lemma~\ref{lem-bound} we can obtain a \cfcn{} coloring of $G$.
We partition the vertices of each clique $C$ in $G \setminus X$ based on their neighborhoods in $X$.
For every subset $Y \subseteq X$, $T_Y^C:=\{x \in C ~|~ N(x) \cap X=Y\}$. 
Notice that in this way we can partition vertices of a clique $C$ into at most $2^{d}$ subsets (called types), 
one for each $Y \subseteq X$.
We represent each clique $C$ in $G \setminus X$ with a vector $T^C$ of length $2^{d}$, where each entry of $T^C$ corresponds to a type and its value equals to the number of vertices in that type.

\begin{Reduction}
For a clique $C \in G \setminus X$, if a type $T_Y^C$ has more than $k+1$ vertices for some $Y \subseteq X$, then removing
all vertices except $k+1$ from $T^C_Y$ does not change $ \chi_{cf} [G]$. 
\end{Reduction}

\begin{proof}
Let $G_1$ be the graph obtained after applying the reduction rule $1$ on $G$.
Given a \cfcn{} coloring $C_{G_1}$ of reduced instance $G_1$, we extend it to a \cfcn{} coloring $C_G$ of $G$.
Let $C$ be a clique in $G \setminus X$ such that the number of vertices in $T_Y^C$ are more than $k+1$ for some $Y \subseteq X$.
We color the  deleted vertices in $T_Y^C$ as follows.
Since there are at least $k+1$ vertices in $T_Y^C$ after applying reduction rule, 
there exists two vertices $u$ and $v$ in $T_Y^C$ such that $C_G(v)=C_G(u)$. 
For each deleted vertex of $T_Y^C$ we assign the color $C_G(u)$. 
Since $C_G(u)$ is not unique in $N[v]$ for any $v \in V(G)$, it is easy to see that $C_G$ is a \cfcn{} coloring of $G$.
\end{proof}

After applying Reduction rule 1, each clique in $G_1 \setminus X$ has at most $k+1$ vertices in each type.
Now we partition the cliques in $G_1 \setminus X$ based on their type vector of length $2^d$. 
For every subset $S \subseteq \{0,1, \cdots, k+1\}^{2^d}$, $T_S^{G_1}:=\{ C \in G_1 \setminus X ~|~ T^C=S\}$. 
Notice that in this way we can partition cliques of $G_1 \setminus X$ into at most $({k+2})^{2^{d}}$ subsets (called mega types), 
one for each $S \subseteq \{0,1, \cdots, k+1\}^{2^d}$.

Let $\tau$ be an arbitrary but fixed ordering on $V(G_1)$. For a vertex $v$ in the modulator $X$ and a clique $C$ in $G_1 \setminus X$, we say that $C$ is critical for $v$ with respect to a \cfcn{} coloring if the first uniquely colored vertex in the neighborhood of $v$ belongs to $C$, where the notion of the first vertex is with respect to the ordering $\tau$. 

\begin{Reduction}
For a subset $S \subseteq \{0,1, \cdots, k+1\}^{2^d}$, If a mega type $T_S^{G_1}$ has more than $d+1$ cliques then removing all cliques except $d+1$ cliques from $T_S^{G_1}$ does not change the $ \chi_{cf} [G_1]$.
\end{Reduction}

\begin{proof}
Let $G_2$ be the graph obtained after applying the reduction rule $2$ on $G_1$.
Given a \cfcn{} coloring $C_{G_2}$ of reduced instance, we extend it to a \cfcn{} coloring $C_{G_1}$ of $G_1$.
Let $T_S^{G_1}$ be a mega type with more than $d+1$ cliques for some $S \subseteq \{0,1, \cdots, k+1\}^{2^d}$.
We color the  deleted cliques in $T_S^{G_1}$ as follows. First, for every vertex $v \in X$, we \textit{mark} a clique $C$ in $T_S^{G_2}$ if it is critical for $v$ with respect to $C_{G_2}$.  
Since there are $d+1$ cliques in $T_S^{G_2}$ after applying the reduction rule 2, at the end of the procedure above, at least one clique is not marked. Let this clique be $C$.

Note that reusing the colors of clique $C$ to color deleted cliques does not violate the uniqueness of a color in $N[x]$ for all $x \in X$.
So we recolor all deleted cliques according to the coloring of $C$.
\end{proof}
After applying the above reduction rules on input graph $G$, It is easy to see that the size of the reduced instance is at most $O((k+2)^{2^d}(d+1)(k+1))$. As
$k < d+2$, we get a kernel of size at most  $O(d^{2^d+2})$.

\par {\bf \cfon{} coloring} :
The proof of \cfon{} coloring is similar to \cfcn{} coloring except some minor changes in reduction rules.

\setcounter{Reduction}{0}
\begin{Reduction}
For a clique $C \in G \setminus X$, if a type $T_Y^C$ has more than $2k+1$ vertices for some $Y \subseteq X$, then removing
all vertices except $2k+1$ from $T_Y^C$ does not change $ \chi_{cf} (G)$. 
\end{Reduction}
\begin{proof}
Let $G_1$ be the graph obtained after applying the reduction rule $1$ on $G$.
Given a \cfon{} coloring $C_{G_1}$ of reduced instance $G_1$, we extend it to a \cfon{} coloring $C_G$ of $G$.
Let $C$ be a clique in $G \setminus X$ such that the number of vertices in $T_Y^C$ are more than $2k+1$ for some $Y \subseteq X$.
Given a \cfon{} coloring $C_G$ of reduced instance, we can color the  deleted vertices in $T_Y^C$ as follows.
Since there are at least $2k+1$ vertices in $T_Y^C$ after applying reduction rule, 
there exists three vertices $u$, $v$ and $w$ in $T_Y^C$ such that $C_G(v)=C_G(u)=C_G(w)$. 
For each deleted vertex of $T_Y^C$ we assign the color $C_G(u)$. 
Since $C_G(u)$ is not unique in $N(x)$ for any $x \in G$, 
it is easy to see that $C_G$ is a \cfon{} coloring of $G$.
\end{proof}

 \begin{Reduction}
For a subset $S \subseteq \{0,1, \cdots, k+1\}^{2^d}$, If  mega type $T_S^{G_1}$ has more than $d+1$ cliques then removing all cliques except $d+1$ cliques does not change $ \chi_{cf} (G_1)$.
\end{Reduction}

After applying the above two reduction rules on input graph $G$ we obtain a kernel of size at most 
$O(d^{2^d+2})$ for \cfon{} coloring problem.   \qed
\end{proof}

\subsection{Parameterized by Distance to Threshold Graphs}
In this section we obtain an additive one (resp. two) approximation algorithm for \cfcn{} (resp. \cfon{}) coloring of a graph parameterized by distance to threshold graphs in $\fpt$-time. 
\begin{theorem}\label{th-threshold}
Let $G$ be a graph and $X \subseteq V(G)$ of size $d$ such that $G \setminus X$ is a threshold graph. Then we can find a \cfcn{} $k$-coloring of $G$ such that  $\chi_{cf}[G] \leq k \leq \chi_{cf}[G]+1$ in $\fpt$-time.
\end{theorem}
\begin{proof}
Let $G'=G[X \cup N(X)]$ is a subgraph of $G$ induced by $X$ and its neighbors. Let $H$ be a graph obtained from 
$G'$ by removing all edges between vertices of $N(X)$. The set $X$ is a vertex cover of $H$ of size $d$.
We divide the problem into two subproblems. First, partial \cfcn{} coloring of $H$: color the vertices of $H$ using the minimum number of colors such that for each $x \in X$, $N[x]$ has a unique color.
Second, color the vertices of $V(G) \setminus V(H)$ such that $N[x]$ has a unique color for each $x \in G \setminus X$.
It is easy to see that $\chi_{cf}[G]$ is at least the number of colors needed in a partial \cfcn{} coloring of $H$.

Now, we give a procedure to find a partial \cfcn{} coloring of $H$. 
We partition the vertices of $N(X)$ based on their neighborhoods in $X$ into at most $2^{d}$ subsets (called types).
For every subset $Y \subseteq X$, $T_Y^H:=\{x \in N(X) ~|~ N(x) \cap X=Y\}$. 
Since every vertex cover of a graph is also a cluster vertex deletion set, therefore
by following the proof of Theorem~\ref{th-dcluster} we get a partial \cfcn{} coloring of $H$. 

Color all vertices of $(G \setminus X) \setminus (N(X))$ with any arbitrary color used in $H$. Since these vertices are 
non-neighbors of $X$, this step does not disturb the existence of unique color in $N[x]$ for all $x \in X$. 
Recall that every connected threshold graph has a universal vertex. 
Recolor the universal vertex $u \in G \setminus X$ with a new color which is not used in $G$. This new color is unique color in 
$N[x]$ for all $x \in G \setminus X$. This completes the proof of \cfcn{} coloring. \qed
\end{proof}

\begin{corollary}
Let $G$ be a graph and $X \subseteq V(G)$ of size $d$ such that $G \setminus X$ is a threshold graph. Then we can find a \cfon{} $k$-coloring of $G$ such that  $\chi_{cf}(G) \leq k \leq \chi_{cf}(G)+2$ in $\fpt$-time.
\end{corollary}
\begin{proof}
The proof of \cfon{} coloring is similar to Theorem~\ref{th-threshold}, except at the end we need to recolor two vertices, universal vertex and some arbitrary vertex of $G \setminus X$ with two new colors. Hence we get a \cfon{} $k$-coloring of $G$ such that  $\chi_{cf}(G) \leq k \leq \chi_{cf}(G)+2$. 
\end{proof}

\section{Special Graph Classes}
In this section we study the complexity of \cfcn{} and \cfon{} coloring problems for bipartite graphs, split graphs, interval graphs and cographs. For a bipartite graph $G=(A,B,E)$, the \cfcn{} coloring can be obtained by coloring all vertices of $A$ with color $0$ and $B$ with $1$. The \cfon{} coloring for bipartite graphs is $\NP$-complete (follows from Theorem~3 in \cite{gargano2015complexity}). We show that \cfcn{} coloring can be solved in polynomial time on split graphs, whereas \cfon{} is coloring $\NP$-complete. For cographs we prove that both \cfcn{} and \cfon{} coloring problems can be solved in polynomial time. For interval graphs we give upper bounds on the number of colors needed for both \cfcn{} and \cfon{} colorings. 
\subsection{Split Graphs}
We begin by showing that \cfcn{} coloring can be solved in polynomial time on split graphs. 
\begin{lemma}\label{lem-cfcn-upper}
Let $G$ be a split graph with at least one edge, then $$2 \leq \chi_{cf}[G] \leq 3$$ 
\end{lemma}
\begin{proof}
Let $G=(C,I)$ be a split graph. We assume that $|C| \geq 3$ and $|I| \geq 3$. Clearly $2 \leq \chi_{cf}[G]$ and we obtain a \cfcn{} 3-coloring $C_G$ of $G$ as follows. 
For a vertex $v \in C$,  assign $C_G(v)=0$ and for all $u \in C \setminus \{v\}$ assign $C_G(u)=1$.
Color all vertices of independent set $I$ with $2$. It is easy to see  that $C_G$ is a \cfcn{} coloring of $G$.
Therefore we have $2 \leq \chi_{cf}[G] \leq 3$. \qed
\end{proof}
In the following lemma we give a characterization of split graphs that admit a \cfcn{} coloring using two colors.
If there is a universal vertex in the graph then coloring it with one color and the rest with second color gives a \cfcn{} coloring.
We assume that graph does not contain universal vertices. 
\begin{lemma}\label{lem-cfcn-split}
 Let $G=(C,I)$ be a split graph without universal vertex. $\chi_{cf}[G]=2$ if and only if for all $v \in C$, $|N(v) \cap I|=1$.
\end{lemma}
\begin{proof}
For the reverse direction, assume for all $v \in C$, $|N(v) \cap I|=1$ then color all vertices of $C$ with $0$ and $I$ with $1$ respectively to get a \cfcn{} coloring of $G$ with two colors.

For the forward direction, assume $\chi_{cf}[G]=2$. 
First, we show that all vertices in the clique get the same color. 
Suppose not, then at least one vertex $v$ in the clique has another color.
If $N(v) \cap I = \emptyset$ and let $v'\in C$ such that $u' \in N(v') \cap I$ 
then color of $u'$ have to be different from $v'$, then $v'$ don not have unique color in its neighborhood.
If $ u \in N(v) \cap I$ then since $v$ is not universal vertex there is a vertex $u'$ in $I$ not adjacent to $v$. 
Let $v'$ be the neighbor of $u'$ in $C$. Then it is easy to see that either $u'$ or $v'$ have no uniquely colored vertex in any coloring. 

If all vertices in the clique have the same color, then all vertices in $I$ must have a different color as otherwise their closed neighborhood would be in conflict. If $C$ and $I$ are completely colored as described above, then every vertex in $C$  must have exactly one neighbor in $I$ to make its neighborhood conflict-free. This completes the proof of the lemma.
\end{proof}
\begin{theorem}
The \cfcn{} coloring problem is polynomial time solvable on the class of split graphs.
\end{theorem}
\begin{proof}
Follows from Lemma~\ref{lem-cfcn-upper} and Lemma~\ref{lem-cfcn-split}.
\end{proof}

Now we show that \cfon{} coloring problem is $\NP$-complete for split graphs. 
\begin{theorem}
The \cfon{} coloring is $\NP$-complete on the class of split graphs.
\end{theorem}

\begin{proof}
We give a reduction from the well-known $\NP$-complete \textsc{Graph Coloring} problem~\footnote{assignment of colors to the vertices of a graph such that no two adjacent vertices have the same color.}.
Given an instance  $(G,k)$ of graph coloring, define a graph $G'$ with $V(G')=V(G) \cup \{x,y\}$ and $E(G') = E(G) \cup  \{xy\} \cup \{xv, yv ~|~ v \in V(G)\}$.
i.e., $x$ and $y$ are universal vertices in graph $G'$.
Construct a split graph $H=(C,I)$ as follows.  

$$ V(H)= V(G') \cup \{I_{uv} ~|~ uv \in E(G')\} \text \qquad { and } $$ 
$$E(H)=  ~\{uv ~|~ u,v \in V(G')\} ~\cup ~\{uI_{uv},vI_{uv}\ ~|~ uv \in E(G') \}$$

It is easy to see that the graph $H$ is a split graph with clique $C=H[V(G')]$ and independent set $I=H[V(H) \setminus V(G')]$.
The construction of graph $H$ can be done in polynomial time. 
We show that for any $k \geq 3$, the graph $G$ is $k$-colorable if and only if the split graph $H$ has a \cfon{} $(k+2)$-coloring.

Let $C_G$ be a $k$-coloring of $G$.
We construct a \cfon{} $(k+2)$-coloring $C_H$ of $H$ as follows.
$C_H(v)=C_G(v)$ for all $v \in V(G)$, $C_H(x)=k+1$, $C_H(y)=k+2$ and color all independent set vertices of 
$H$ with color $k$.
Now we show that $C_H$ is a \cfon{} $(k+2)$-coloring of $H$.
According to the coloring $C_H$, the unique color in $N(v)$ is $k+2$ if $v \in V(G) \cup \{x\}$ and $k+1$ if $v=y$.   
For all $I_{uv} \in V(H)$, $N(I_{uv})=\{u,v\}$ and we have $C_H(u) =C_G(u)\neq C_G(v)=C_H(v)$ for all $uv \in E(G')$, therefore 
$N(I_{uv})$ has a unique color.

For the reverse direction, Let $C_H$ be a \cfon{} $(k+2)$-coloring of $H$, 
we show that $C_H$ when restricted to vertices of $G$ gives a $k$-coloring of $G$.  
By the construction, for any $I_{uv} \in V(H)$, we have $N( I_{uv}) = \{ u , v \}$,  therefore $C_H(u) \neq C_H(v)$ which implies
$C_H$ when restricted to vertices of $G'$ gives a proper coloring of $G'$ with $k+2$ colors.

Now we show that the color of $x$ (resp $y$) is unique in $V(G')$.
For every $v \in V(G)$ we have $N(I_{vx})=\{v,x\}$, implies $C_H(x) \neq C_H(v)$ for all $v \in V(G)$.
Similarly $C_H(y) \neq C_H(v)$ for all $v \in V(G)$ and $C_H(x) \neq C_H(y)$ as $N( I_{xy}) = \{ x , y \}$.
Therefore $C_H$ when restricted to $V(G)$ gives a $k$-coloring of $G$. \qed
\end{proof}
\subsection{Interval Graphs}
In this section, we give upper bounds on the number of colors needed for both \cfcn{} and \cfon{} coloring problems for interval graphs. 
Let $G$ be an interval graph and $\mathcal{I}$ be an \emph{interval representation} of $G$, i.e., 
there is
a mapping from $V(G)$ to closed intervals on the real line such that for any
two vertices $u$ and $v$, $uv \in E(G)$ if and only if $I_u \cap I_v \neq \emptyset$.
For any interval
graph, there exists an interval representation with all endpoints distinct. Such
a representation is called a \emph{distinguishing interval representation} and it can be computed starting from an
arbitrary interval representation of the graph. 
Interval graphs can be recognized in linear time and an interval representation can be obtained in linear time \cite{booth1976testing}.
 Let $l(I_u)$ and $r(I_u)$ denote the left and right end points of the interval corresponding
to the vertex $u$ respectively. We say that an interval $I \in \mathcal{I}$ is \emph{rightmost interval} if $r(J) \leq r(I)$ for all $J \in \mathcal{I}$.

\begin{algorithm}[t]
\caption{Conflict-free closed neighborhood coloring of an  interval graph with at most four colors}
\label{algo-cfcn}
\KwIn{A connected interval graph $G$ along with its distinguishing interval representation $\mathcal{I}$.}
\KwOut{A conflict-free closed neighborhood coloring $C_G$ of $G$ with at most four colors.}
\For{{\bf each} interval $I_{v_i} \in \mathcal{I}$ by increasing left end point}{
\If {$v_i$ is not colored}{ 
\eIf {$I_{v_i}$ is the rightmost interval in $\mathcal{I}$}{
$C_G(v_i)=1$ and $C_G(v_j)=0$ for all $v_j \in N(v_i)$  with $l(I_{v_j}) \geq l(I_{v_{i}})$  \;
}
{find  $v_l \in N(v_i)$ such that $r(I_{v_j}) \leq r(I_{v_l})$ for all $v_j \in N[v_i]$\; 
\eIf {$I_{v_l}$ is the rightmost interval in $\mathcal{I}$}{
 $C_G(v_i)=1$, $C_G(v_l)=2$ and  $C_G(v_j)=0$ for all $v_j \in N(v_i \cup v_l)$ with $l(I_{v_j}) \geq l(I_{v_{i}})$\;
}
{find $v_{l'}\in N(v_l)$ such that $r(I_{v_j}) \leq r(I_{v_{l'}})$ for all $v_j \in N[v_l]$ \;
$C_G(v_i)=1, C_G(v_l)=2$, $C_G(v_{l'})=3$ and $C_G(v_j)=0$ for all $v_j \in N(v_i \cup v_l \cup v_{l'})$ with $r(I_{v_j}) \leq r(I_{v_{l'}})$ and $l(I_{v_j}) \geq l(I_{v_{i}})$\;}

{\bf end if }

}
{\bf end if }

}
{\bf end if }
  
}
{\bf end for }

\end{algorithm}

\setcounter{theorem}{4}

\begin{theorem}\label{th-cfcn-interval}
Let $G$ be an interval graph with at least one edge, then $$2 \leq \chi_{cf}[G] \leq 4$$
\end{theorem}
\begin{proof}
It is easy to see that for any graph with at least one edge, $2 \leq \chi_{cf}[G]$.
The algorithm that computes the \cfcn{} coloring of $G$ with at most four colors is given in Algorithm~\ref{algo-cfcn}.
Now we present details on the
proof of correctness of the algorithm.

Let $C_G$ be a \cfcn{} coloring of $G$  obtained using Algorithm~\ref{algo-cfcn}. 
For any $u_i, u_j \in V(G)$, if $C_G(u_i)=C_G(u_j)=1$ then $u_iu_j \notin E(G)$:
Suppose assume that $u_iu_j \in E(G)$ and $l(I_{u_i}) \leq l(I_{u_j})$. 
If $I_{u_j}$ is the rightmost interval in $\mathcal{I}$ then $C_G(u_j)=2$ [lines 7-8 in Algorithm~\ref{algo-cfcn}],
which is a contradiction to $C_G(u_2)=1$.
If $I_{u_j}$ is not the rightmost interval in $\mathcal{I}$,  
then there exists $v_l$ adjacent to $u_i$ such that 
$r(I_{v_k}) \leq r(I_{v_l})$ for all $v_k \in N(u_i)$. So we have $r(I_{u_j}) \leq r(I_{v_l})$, and the algorithm colors $u_j$ with $0$, which is a contradiction as $C_G(u_j)=1$. 
If $I_{u_j}$ lies completely inside $I_{u_i}$ then $C_G(u_j)=0$, which is again a contradiction to $C_G(u_j)=1$.
Therefore $u_iu_j \notin E(G)$.
Similarly, we can show that any two vertices of color $2$ or color $3$ are also not adjacent. This shows that any vertex $v$ colored with a non-zero color by Algorithm~\ref{algo-cfcn} have a unique color in the set $N[v]$.

Now we show that if a vertex $v$ is assigned a color $0$ by Algorithm~\ref{algo-cfcn} then $v$ has at least one 
non-zero unique color in $N[v]$. It is easy to see that if a vertex is colored $0$ by Algorithm~\ref{algo-cfcn} then it is adjacent to at least one vertex of non-zero color. 
Since $C_G(v)=0$, $v$ is not the rightmost interval in $\mathcal{I}$. 
Suppose assume that $v$ is adjacent to two vertices $u_1$ and $u_2$ of color $1$ and $l(I_{u_1}) \leq l(I_{u_2})$. 
Note that $u_1$ and $u_2$ are not adjacent as $C_G(u_1)=C_G(u_2)= 1$, so we have  $r(I_{u_1}) \leq l(I_{u_2})$.

Since $C_G(v)=0$ there exists a $v_l$ in $N(u_1)$ such that 
 $r(I_{v_j}) \leq r(I_{v_l})$ for all $v_j \in N[u_1]$. This implies $r(I_{v}) \leq r(I_{v_l})$ and $I_{v_l} \cap I_{u_2} \neq \emptyset$. 
If $u_2$ is the rightmost interval, then $C_G(v_l)=2$ and $C_G(u_2)=3$ which is a contradiction to $C_G(u_2)=1$. 
If $u_2$ is not the rightmost interval, then
there exists a $v_{l'}$ in $N(v_l)$ such that 
$r(I_{v_j}) \leq r(I_{v_{l'}})$ for all $v_j \in N[v_l]$. This implies $r(I_{u_2}) \leq r(I_{v_{l'}})$.
The algorithm colors vertex $v_l$ with $2$, $v_{l'}$ with $3$ and $u_2$ with $0$, which is again a contradiction to $C_G(u_2)=1$. 
 
Along similar lines, we can also show that any vertex colored zero can not be adjacent to two vertices of same non-zero color.\qed

\end{proof}

\begin{theorem}
Let $G$ be an interval graph with at least two edges, then $$2 \leq \chi_{cf}(G) \leq 4$$
\end{theorem}
\begin{proof}
The algorithm that finds a \cfon{} coloring of interval graphs with four colors is given in Algorithm~\ref{algo-cfon}.
The correctness proof of Algorithm~\ref{algo-cfon} is similar to proof of Theorem~\ref{th-cfcn-interval}.   \qed
\end{proof}

\begin{algorithm}[t]
\caption{Conflict-free open neighborhood coloring of an interval graph with at most four colors}
 \label{algo-cfon}
\KwIn{A connected interval graph $G$ along with its distinguishing interval representation $\mathcal{I}$.}
\KwOut{A conflict-free open neighborhood coloring $C_G$ of $G$ with four colors.}
\For{{\bf each} interval $I_{v_i} \in \mathcal{I}$ by increasing left end point}{
\If {$v_i$ is not colored}{ 
\eIf {$I_{v_i}$ is the rightmost interval in $\mathcal{I}$}{
\eIf{there is no vertex $v_{i'}$ such that $I_{v_{i'}} \subseteq I_{v_i}$ }{
$C_G(v_i)=1$  \;}
{select an arbitrary vertex $v_{i'}$ in $N(v_i)$ with $l(I_{v_{i'}}) \geq l(I_{v_{i}})$ \; 
$C_G(v_i)=1$,$C_G(v_{i'})=2$ and $C_G(v_j)=0$ for all $v_j \in N(v_i)$ with $l(I_{v_j}) \geq l(I_{v_{i}})$}
{\bf end if }
}
{find a $v_l \in N(v_i)$ such that $r(I_{v_j}) \leq r(I_{v_l})$ for all $v_j \in N[v_i]$\; 
\eIf {$I_{v_l}$ is the rightmost interval in $\mathcal{I}$}{
 $C_G(v_i)=1$, $C_G(v_l)=2$ and  $C_G(v_j)=0$ for all $v_j \in N(v_l \cup v_i)$ with $l(I_{v_j}) \geq l(I_{v_{i}})$  \;
}
{find a $v_{l'}\in N(v_l)$ such that $r(I_{v_j}) \leq r(I_{v_{l'}})$ for all $v_j \in N[v_l]$ \;
$C_G(v_i)=1$,$C_G(v_l)=2$, $C_G(v_{l'})=3$ and $C_G(v_j)=0$ for all $v_j \in N(v_i \cup v_l \cup v_{l'})$ with $r(I_{v_j}) \leq r(I_{v_{l'}})$ and $l(I_{v_j}) \geq l(I_{v_{i}})$ \; }

{\bf end if }

}
{\bf end if }

}
{\bf end if }
  
}
{\bf end for }

\end{algorithm}


\subsection{Cographs}
\begin{theorem}\label{th-cographs}
The \cfcn{} (resp. \cfon{}) coloring problem can be solved in polynomial time on cographs.
\end{theorem}
\begin{proof}
We use modular decomposition \cite{habib2010survey} technique to solve conflict free coloring problem on cographs. 
First, we define the notion of \emph{modular decomposition}.
A set $M \subseteq V(G)$ is called {\it module} of $G$ if all vertices of $M$ have the same set 
of neighbors in $V(G)\setminus M$.  
The \emph{trivial modules} are $V(G)$, and $\{v\}$ for all $v$.
A prime  graph is a graph in which all
modules are trivial.
The modular decomposition of a graph is one of the decomposition techniques which was introduced by Gallai~\cite{gallai1967transitiv}. The {\it modular decomposition} of a graph $G$ is a rooted tree $M_G$ 
that has the following properties:
\begin{enumerate}
\setlength{\itemsep}{1pt}
\setlength{\parskip}{0pt}
 \item The leaves of $M_G$ are the vertices of $G$.
 \item For an internal node $h$ of $M_G$, let $M(h)$ be the set of vertices of $G$ that are leaves of the subtree of 
 $M_G$ rooted at $h$. ($M(h)$ forms a module in $G$).
 \item For each internal node $h$ of $M_G$ there is a graph $G_h$ (\emph{representative graph}) with 
 $V(G_h)=\{h_1,h_2,\cdots,h_r\}$, where $h_1,h_2,\cdots,h_r$ are the 
 children of $h$ in $M_G$ and for  $1 \leq i<j \leq r$, $h_i$ and $h_j$ are adjacent in $G_h$ iff there are vertices $u \in M(h_i)$ and $v \in M(h_j)$ that are adjacent in $G$.
 \item $G_h$ is either a clique, an independent set, or a prime graph and 
 $h$ is labeled \emph{Series} if $G_h$ is clique, \emph{Parallel} if $G_h$ is an independent set, and \emph{Prime} otherwise.
\end{enumerate}

Modular decomposition tree of cographs has only parallel and series nodes. 
Let $G$ be a cograph whose modular decomposition tree is $M_G$. Without loss of generality we assume that 
the root $r$ of tree $M_G$ is a series node, otherwise, $G$ is not connected and we can color each 
connected component independently. 
Let the children of $r$ be $x$ and $y$. Further, let the cographs corresponding to the subtrees at 
$x$ and $y$ be $G_x$ and $G_y$. 
First we consider \cfcn{} coloring problem on cographs. If $G$ has a universal vertex then color it with one color and the rest with a second color. If $G$ does not have a universal vertex, then color one vertex of $G_x$ with $0$, one vertex of $G_y$ with $1$ and all other remaining vertices of $G$ with $2$.

For the \cfon{} coloring, if $G$ contains only two vertices then color one vertex with color $0$ and other one with color $1$. If $G$ contains at least three vertices then color one vertex of $G_x$ with $0$, one vertex of $G_y$ with $1$ and all other remaining vertices of $G$ with $2$. \qed
\end{proof}

\section{Conclusion}
In this paper, we have studied the parameterized complexity of conflict-free coloring problem with respect to open/closed neighborhoods.
We have shown that both closed and open neighborhood conflict-free colorings are $\fpt$ parameterized by cluster vertex deletion number and also showed that both variants of the problem admit an additive constant approximation algorithm when parameterized by the distance to threshold graphs. We studied the complexity of the problem on special classes of graphs. We show that both closed and open neighborhood conflict-free colorings are polynomially solvable on cographs. On split graphs, closed neighborhood coloring can be solved in polynomial time, whereas open neighborhood coloring is $\NP$-complete. For interval graphs, we give upper bounds on the number of colors needed for both open/closed conflict-free coloring problems. 

The following problems remain open. 
\begin{itemize}
 \item Does conflict-free open/closed coloring admit a polynomial kernel when parameterized by (a) the size of a vertex cover (b) distance to clique?
 \item Is the \cfcn{} problem is \fpt{} when parameterized by distance to cographs, distance to split graphs?
 \item What is the parameterized complexity of both the problems when parameterized by clique-width/rank-width?
 \item  What is the complexity of \cfcn{} and \cfon{} coloring problems on interval graphs?
\end{itemize}


\bibliographystyle{splncs03}
\bibliography{cfc.bib}

\end{document}